\newif\ifdraft\drafttrue
\newif\ifcolor\colortrue

\documentclass[preprint]{sig-alternate-10pt}

\usepackage{xcolor}
\definecolor{pennblue}{cmyk}{1,.65,0,.3}
\definecolor{pennred}{cmyk}{0,1,.65,.34}
\usepackage{cmtt}

\usepackage{times}
\usepackage{tikz}
\usepackage{alltt}
\usepackage{graphicx}
\usepackage{textgreek}
\usepackage{balance}
\usepackage{amsmath}
\usepackage{pifont}
\usepackage{color}
\usepackage{colortbl}
\usepackage{local}
\usepackage{multirow}

\definecolor{williamspurple}{RGB}{89,23,128}
\definecolor{williamsgold}{RGB}{255,213,0}

\newcommand{\statement}[3]{#1 : #2 \rightarrow #3}

\newcommand{\PMatch}[3]{\ensuremath{#1.#2=#3}}

\newcommand{\true}{\kw{true}}
\newcommand{\false}{\kw{false}}

\newcommand{\MyAnd}[2]{#1 \ \kw{and} \ #2}
\newcommand{\Or}[2]{#1 \ \kw{or} \ #2}
\newcommand{\Not}[1]{!~#1}

\newcommand{\Max}[2]{\kw{max}(#1,#2)}
\newcommand{\Min}[2]{\kw{min}(#1,#2)}

\ifcolor
\usepackage[colorlinks=true,linkcolor={pennblue},citecolor={pennred},urlcolor={pennblue}]{hyperref}
\else
\usepackage[colorlinks=true,linkcolor={black},citecolor={black},urlcolor={black}]{hyperref}
\fi
\usepackage{url}

\definecolor{Gray}{gray}{0.9}

\newcommand{\superscript}[1]{\ensuremath{^{\textrm{#1}}}}
\def\wu{\superscript{*}}
\def\wg{\superscript{\dag}}

\begin{document}
%
\conferenceinfo{WOODSTOCK}{'97 El Paso, Texas USA}

\title{ Merlin: A Language for Provisioning Network Resources }

\numberofauthors{1} 

\author{
{\large
Robert Soul\'{e}{\wu} \quad
Shrutarshi Basu{\wg} \quad
Parisa Jalili Marandi{\wu} \quad
Fernando Pedone{\wu}} \\[.25em]
{\large
Robert Kleinberg{\wg} \quad
Emin G\"{u}n Sirer{\wg} \quad
Nate Foster{\wg}}\\[.25em]
\large {\wu}University of Lugano \quad
\large {\wg}Cornell University\\[.25em]
}

\date{}

\maketitle
\begin{abstract}
This paper presents Merlin, a new framework for managing resources in software-defined networks. With Merlin, administrators express high-level policies using programs in a declarative language. The language includes logical predicates to identify sets of packets, regular expressions to encode forwarding paths, and arithmetic formulas to specify bandwidth constraints. The Merlin compiler uses a combination of advanced techniques to translate these policies into code that can be executed on network elements including a constraint solver that allocates bandwidth using parameterizable heuristics. To facilitate dynamic adaptation, Merlin provides mechanisms for delegating control of sub-policies and for verifying that modifications made to sub-policies do not violate global constraints. Experiments demonstrate the expressiveness and scalability of Merlin on real-world topologies and applications. Overall, Merlin simplifies network administration by providing high-level abstractions for specifying network policies and scalable infrastructure for enforcing them.

\end{abstract}





\section{Introduction}
\label{sec:introduction}

Network operators today must deal with a wide range of management
challenges from increasingly complex policies to a proliferation of
heterogeneous devices to ever-growing traffic demands.
Software-defined networking (SDN) provides tools that could be used to
address these challenges, but existing SDN APIs are either too
low-level or too limited in functionality to enable effective
implementation of rich network-wide policies. As a result, there is
widespread interest in academia and industry in higher-level
programming languages and ``northbound APIs'' that provide convenient
control over the complete set of resources in a network.

Unfortunately, despite several notable advances, there is still a
large gap between the capabilities of existing SDN APIs and the
realities of network management. Current programming languages focus
mostly on packet forwarding and largely ignore functionality
such as bandwidth and richer packet-processing functions that can only
be implemented on middleboxes, end hosts, or with custom
hardware~\cite{foster11,monsanto13,voellmy13,anderson14,flowlog14}. Network
orchestration frameworks provide powerful mechanisms that handle a
larger set of concerns including middlebox placement and
bandwidth~\cite{gember12,joseph08,qazi13,sekar12}, but they expose
APIs that are either extremely simple (e.g., sequences of middleboxes)
or not specified in detail. Overall, the challenges of managing
real-world networks using existing SDN APIs remain unmet.

This paper presents a new SDN programming language that is designed to
fill this gap. This language, called Merlin, provides a collection of
high-level programming constructs for (i) classifying packets; (ii)
controlling forwarding paths; (iii) specifying rich packet
transformations; and (iv) provisioning bandwidth in terms of maximum
limits and minimum guarantees. These features go far beyond what can
be realized just using SDN switches or with existing languages like
Frenetic~\cite{foster11}, Pyretic~\cite{monsanto12}, and
Maple~\cite{voellmy13}. As a result, implementing Merlin is
non-trivial because it involves determining allocations of limited
network-wide resources such as bandwidth---the simple compositional
translations used in existing SDN compilers cannot be readily extended
to handle the new features provided in Merlin.

Merlin uses a variety of compilation techniques to determine
forwarding paths, map packet transformations to network elements, and
allocate bandwidth. These techniques are based on a unified
representation of the physical network topology and the constraints
expressed by the policy---a logical topology. For traffic with
bandwidth guarantees, the compiler uses a mixed integer program
formulation to solve a variant of the multi-commodity flow constraint
problem. For traffic without bandwidth guarantees, Merlin leverages
properties of regular expressions and finite automata to efficiently
generate forwarding trees that respect the path constraints encoded in
the logical topology. Handling these two kinds of traffic separately
allows Merlin to provide a uniform interface to programmers while
reducing the size of the more expensive constraint problems that must
be solved. The compiler also handles generation of low-level
instructions for a variety of elements including switches,
middleboxes, and end hosts.


Although the configurations emitted by the Merlin compiler are static,
the system also incorporates mechanisms for handling dynamically
changing policies. Run-time components called \emph{negotiators}
communicate among themselves to dynamically adjust bandwidth
allocations and \emph{verify} that the modifications made by other
negotiators do not lead to policy violations. Again, the use of a
high-level programming language is essential, as it provides a
concrete basis for analyzing, transforming, and verifying policies.

We have built a working prototype of Merlin, and used it to implement
a variety of practical policies that demonstrate the expressiveness of
the language. These examples illustrate how Merlin supports a wide
range of network functionality including simple forwarding policies,
policies that require rich transformations usually implemented on
middleboxes such as load balancing and deep-packet inspection, and
policies that provide bandwidth guarantees. We have also implemented
negotiators that realize max-min fair sharing and additive-increase
multiplicative-decrease dynamic allocation schemes. Our experimental
evaluation shows that the Merlin compiler can quickly provision and
configure real-world datacenter and enterprise networks, and that
Merlin can be used to obtain better application performance for data
analytics and replication systems.

Overall, this paper makes the following contributions:
\begin{itemize}
\item The design of high-level network management abstractions realized
in an expressive policy language that models packet classification,
forwarding, transformation, and bandwidth.

\item Compilation algorithms based on a translation from
policies to constraint problems that can be solved using mixed integer
programming.

\item An approach for dynamically adapting policies using negotiators 
and accompanying verification techniques, made possible by the
language design.

\end{itemize}
The rest of this paper describes the design of the Merlin language
(\S\ref{sec:language}), compiler (\S\ref{sec:compiler}), and runtime
transformations (\S\ref{sec:transformations}). It then describes the
implementation (\S\ref{sec:implementation}) and presents the results
from our performance evaluation (\S\ref{sec:evaluation}).

\section{Language Design}
\label{sec:language}

\begin{figure}[t]
\(\begin{array}{r@{\;}c@{\;}l@{\quad}l}
\mathit{loc} & \in & \textit{Locations}\\
t & \in & \textit{Packet transformations}\\
pol & \bnf & [s_1; \dots; s_n], \phi & \text{Policies}\\ 
s & \bnf & \statement{id}{p}{r} & \text{Statements}\\ 
\phi & \bnf &\Max{e}{n} \mid \Min{e}{n}   & \text{Presburger Formulas}\\ 
& \mid & \MyAnd{\phi _1}{\phi _2} \mid \Or{\phi _1}{v_2} \mid
\,\Not{\phi _1}  & \\ 
e & \bnf & n \mid id \mid e + e &  \text{Bandwidth Terms}\\
a & \bnf & . \mid c \mid a\, a \mid a \texttt{|} a \mid a^* \mid \,\Not{a} &
\text{Path Expression}\\
p & \bnf &  m \mid \MyAnd{p_1}{p_2} \mid \Or{p_1}{p_2} \mid \,\Not{p_1}   & \text{Predicates}\\
  & \mid & \PMatch{h}{f}{n}  \mid  \true \mid \false & \\
c & \bnf & \mathit{loc} \mid t & \text{Path Element}\\
\end{array} 
\)
\caption{Merlin abstract syntax.}
\label{fig:syntax}
\end{figure}

The Merlin policy language is designed to give programmers a rich
collection of constructs that allow them to specify the intended
behavior of the network at a high level of abstraction. As an example,
suppose that we want to place a bandwidth cap on \textsc{ftp} control
and data transfer traffic, while providing a bandwidth guarantee to
\textsc{http} traffic. The program below realizes this specification
using a sequence of Merlin policy statements, followed by a logical
formula. Each statement contains a variable that tracks the amount of
bandwidth used by packets processed with that statement, a predicate
on packet headers that identifies a set of packets, and a regular
expression that describes a set of forwarding paths through the
network:
\begin{alltt}\small
  [ x : (eth.src = 00:00:00:00:00:01 \texttt{and} 
         eth.dst = 00:00:00:00:00:02 \texttt{and} 
         tcp.dst = 20) -> .* dpi .*
    y : (eth.src = 00:00:00:00:00:01 \texttt{and} 
         eth.dst = 00:00:00:00:00:02 \texttt{and} 
         tcp.dst = 21) -> .*
    z : (eth.src = 00:00:00:00:00:01 \texttt{and} 
         eth.dst = 00:00:00:00:00:02 \texttt{and} 
         tcp.dst = 80) -> .* dpi *. nat .* ],
    max(x + y,50MB/s) and min(z,100MB/s) 
\end{alltt}
The statement on the first line asserts that
\textsc{ftp} traffic from the host with \textsc{mac} address \mbox{{\tt \small
00:00:00:00:00:01}} to the host with \textsc{mac}
address \mbox{{\tt \small 00:00:00:00:00:02}} must travel along a path
that includes a packet processing function that performs deep-packet
inspection (\texttt{\small dpi}). The next two statements identify and
constrain \textsc{ftp} control and \textsc{http} traffic between the
same hosts respectively. Note that the statement for \textsc{ftp}
control traffic does not include any constraints on its forwarding
path, while the \textsc{http} statement includes both a \texttt{\small
dpi} and a \texttt{\small nat} constraint. The formula on the last
line declares a bandwidth cap (\texttt{\small max}) on
the \textsc{ftp} traffic, and a bandwidth guarantee (\texttt{\small
min}) for the \textsc{http} traffic. The rest of this section
describes the constructs used in this policy in detail.

\subsection{Syntax and semantics}

The syntax of the Merlin policy language is defined by the grammar in
Figure~\ref{fig:syntax}. A policy is a set of
\emph{statements}, each of which specifies the handling of a subset of traffic, 
together with a \emph{logical formula} that expresses a global
bandwidth constraint. For simplicity, we require that the statements
have disjoint predicates and together match all packets. In our
implementation, these requirements are enforced by a simple
pre-processor. Each policy statement comprises several components:
an \emph{identifier}, a \emph{logical predicate}, and a \emph{regular
expression}. The identifier provides a way to identify the set of
packets matching the predicate, while the regular expression specifies
the forwarding paths and packet transformations that should be applied
to matching packets. Together, these abstractions facilitate thinking
of the network as a ``big switch''~\cite{kang13}, while enabling
programmers to retain precise control over forwarding paths and
bandwidth usage.

\paragraph*{Logical predicates}
Merlin supports a rich predicate language for classifying packets.
Atomic predicates of the form $\PMatch{h}{f}{n}$ denote the set of
packets whose header field $h.f$ is equal to $n$. For instance, in the
example policy above, statement \texttt{\small z} contains the
predicate that matches packets with \texttt{\small eth} source
address \texttt{\small 00:00:00:00:00:01}, 
destination address 
\texttt{\small 00:00:00:00:00:02}, and \texttt{\small tcp} port \texttt{\small 80}. Merlin provides 
atomic predicates for a number of standard protocols including
Ethernet, \textsc{ip}, \textsc{tcp}, and \textsc{udp}, and a
special predicate for matching packet payloads. Predicates can also be
combined using conjunction (\texttt{\small and}), disjunction
(\texttt{\small or}), and negation (\texttt{\small !}).

\paragraph*{Regular expressions}
Merlin allows programmers to specify the set of allowed forwarding
paths through the network using regular expressions---a natural and
well-studied mathematical formalism for describing paths through a
graph (such as a finite state automaton or a network
topology). However, rather than matching strings of characters, as
with ordinary regular expressions, Merlin regular expressions match
sequences of network locations (including names of transformations, as
described below). The compiler is free to select any matching path for
forwarding traffic as long as the other constraints expressed by the
policy are satisfied. We assume that the set of network locations is
finite.  As with \textsc{posix} regular expressions, the dot symbol
(\texttt{\small .})  matches a single element of the set of all
locations.

\paragraph*{Packet transformations}
Merlin regular expressions may also contain names of packet-processing
functions that may transform the headers and contents of packets. Such
functions can be used to implement a variety of useful operations
including deep packet inspection, network address translation,
wide-area optimizers, caches, proxies, traffic shapers, and
others. The compiler determines the location where each function is
enforced, using a mapping from function names to possible locations
supplied as a parameter. The only requirements on these functions are
that they must take a single packet as input and generate zero or more
packets as output, and they must only access local state. In
particular, the restriction to local state allows the compiler to
freely place functions without having to worry about maintaining
global state. Merlin's notion of packet processing functions that can
be easily moved within the network is similar in spirit to network
function virtualization~\cite{nfv}.

\paragraph*{Bandwidth constraints}
Merlin policies use formulas in Presburger arithmetic (i.e,
first-order formulas with addition but without multiplication, to
ensure decidability) to specify constraints that either limit
(\texttt{\small max}) or guarantee (\texttt{\small min})
bandwidth. The addition operator can be used to specify an aggregate
cap on traffic, such as in the \texttt{\small max(x + y, 50MB/s)} term
from the running example. By convention, policies without a rate
clause are unconstrained---policies that lack a minimum rate are not
guaranteed any bandwidth, and policies that lack a maximum rate may
send traffic at rates up to line speed. Bandwidth limits and
guarantees differ from packet processing functions in one important
aspect: they represent an explicit allocation of global network
resources. Hence, additional care is needed when compiling them.

\paragraph*{Syntactic sugar}
In addition to the core constructs shown in Figure~\ref{fig:syntax},
Merlin also supports several forms of syntactic sugar that can
simplify the expression of complex policies. Notably, Merlin provides
set literals and functions for operating on and iterating over sets. For
example, the following policy,

\begin{alltt}\small
 srcs := \{00:00:00:00:00:01\}
 dsts := \{00:00:00:00:00:02\}
 foreach (s,d) in cross(srcs,dsts):
   tcp.dst = 80 -> 
   ( .* nat .* dpi .*) at max(100MB/s)
\end{alltt}

\noindent
is equivalent to statement \texttt{\small z} from the example. The
sets \texttt{\small srcs} and \texttt{\small dsts} refer to singleton sets of
hosts. The \texttt{\small cross} operator takes the cross product of these
sets. The \texttt{\small foreach} statement iterates over the resulting set,
creating a predicate from the src \texttt{\small s}, destination \texttt{\small d},
and term \texttt{\small tcp.dst = 80}.

\paragraph*{Summary}
Overall, Merlin's policy language enables direct expression of
high-level network policies. Crucial to its design, however, is that
policies can be distilled into their component constructs, which can
then be distributed across many devices in the network to collectively
enforce the global policy. The subsequent sections present these
distribution and enforcement mechanisms in detail.

\section{Compiler}
\label{sec:compiler}

The Merlin compiler performs three essential tasks: (i) it
translates global policies into locally-enforceable policies; (ii) it
determines the paths used to carry traffic across the network, places
packet transformations on middleboxes and end hosts, and allocates
bandwidth to individual flows; and (iii) it generates low-level
instructions for network devices and end hosts. To do this, the
compiler takes as inputs the Merlin policy, a representation of the
physical topology, and a mapping from transformations to possible
placements, and builds a logical topology that incorporates the
structure of the physical topology as well as the constraints encoded
in the policy. It then analyzes the logical topology to determine
allocations of resources and emits low-level configurations for
switches, middleboxes, and end hosts.

\begin{figure}[t]
\centerline{\includegraphics[width=0.9\columnwidth]{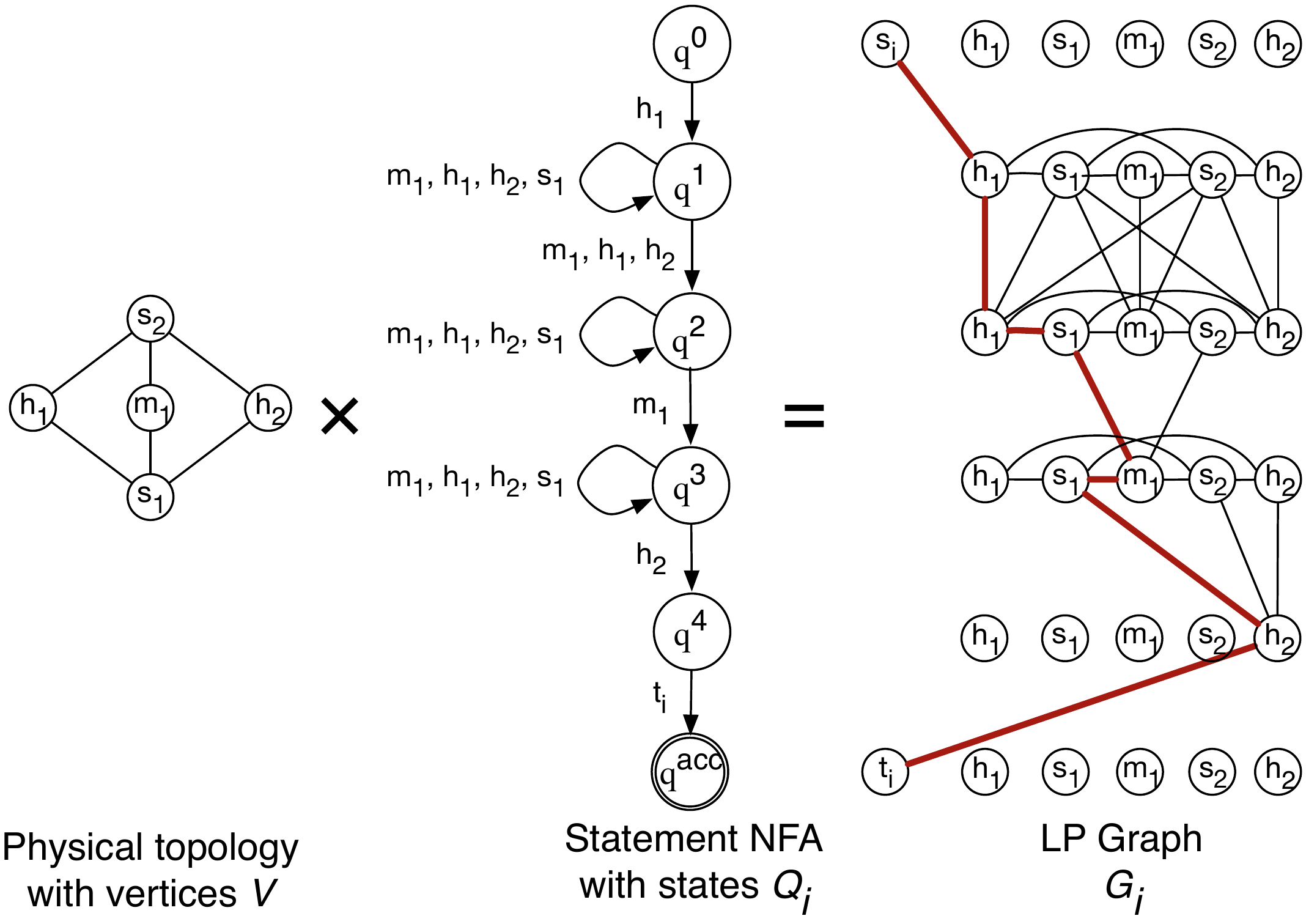}}
\caption{Logical topology for the example
  policy. The thick, red path illustrates a solution.}
\label{fig:ilp}
\end{figure}

\subsection{Localization}

Merlin's Presburger arithmetic formulas are an expressive way to
declare bandwidth constraints, but actually implementing them leads to
several challenges: aggregate guarantees can be enforced using shared
quality-of-service queues on switches, but aggregate bandwidth limits
are more difficult, since they require distributed state in
general. To solve this problem, Merlin adopts a pragmatic
approach. The compiler first rewrites the formula so that the
bandwidth constraints apply to packets at a single location. Given a
formula with one term over $n$ identifiers, the compiler produces a
new formula of $n$ local terms that collectively imply the
original. By default, the compiler divides bandwidth equally among the
local terms, although other schemes are permissible. For example,
the \texttt{\small max} term in the running example would be localized
to \texttt{\small max(x, 25MB/s) and max(y, 25MB/s)}. Rewriting policies in
this way involves an inherent tradeoff: localized enforcement
increases scalability, but risks underutilizing resources if the
static allocations do not reflect actual usage. In
Section~\ref{sec:transformations}, we describe how Merlin navigates
this tradeoff via a run-time mechanism, called \emph{negotiators},
that can dynamically adjust allocations.

\subsection{Provisioning for Guaranteed Rates}
\label{sec:constraint}

The most challenging aspect of the Merlin compilation process is
provisioning bandwidth for traffic with guarantees. To do this, the
compiler encodes the input policy and the network topology into a
constraint problem whose solution, if it exists, can be used to
determine the configuration of each network device.

\newcommand{\grph}{{\mathcal{G}}}
\newcommand{\autom}{{\mathcal{M}}}
\newcommand{\gprod}{{\circ}}
\newcommand{\states}{{\mathcal{Q}}}
\newcommand{\nr}{{r}}
\newcommand{\nrmax}{{\nr_{\mathrm{max}}}}
\newcommand{\unrmax}{{R_{\mathrm{max}}}}
\newcommand{\rmini}{{r^i_{\mathrm{min}}}}

\paragraph*{Logical topology}
Recall that a statement in the Merlin language contains a
regular expression, which constrains the set of forwarding paths and
packet processing functions that may be used to satisfy the
statement. To facilitate the search for routing paths that satisfy
these constraints, the compiler represents them internally as a
directed graph $\grph$ whose paths correspond to physical network
paths that respect the path constraints of a single statement. The
overall graph $\grph$ for the policy is a union of disjoint components
$\grph_i$, one for each statement $i$.

The regular expression $a_i$ in statement $i$ is over the set of
locations and packet processing functions. The first step in the
construction of $\grph_i$ is to map $a_i$ into a regular expression
$\bar{a}_i$ over the set of locations using a simple substition: for
every occurrence of a packet transformation, we substitute the union
of all locations associated with that function. (Recall that the
compiler takes an auxiliary input specifying this mapping from
functions to locations.) For example, if \texttt{\small h1}, \texttt{\small h2},
and \texttt{\small m1} are the three locations capable of running network
address translation, then the regular expression {\tt .* nat .*} would
be transformed to {\tt \small .* (h1|h2|m1) .*}. The next step is to
transform the regular expression $\bar{a}_i$ into a nondeterministic
finite automaton (\textsc{nfa}), denoted $\autom_i$, that accepts the
set of strings in the regular language given by $\bar{a}_i$.

Letting $L$ denote the set of locations in the physical network and
$\states_i$ denote the state set of $\autom_i$, the vertex set of
$\grph_i$ is the Cartesian product $L \times \states_i$ together with
two special vertices $\{s_i,t_i\}$ that serve as a universal source
and sink for paths representing statement $i$ respectively. The graph
$\grph_i$ has an edge from $(u,q)$ to $(v,q')$ if and only if: {\bf
(i)} $u=v$ or $(u,v)$ is an edge of the physical network, and {\bf
(ii)} $(q,q')$ is a valid state transition of $\autom_i$ when
processing $v$. Likewise, there is an edge from $s_i$ to $(v,q')$ if
and only if $(q^0,q')$ is a valid state transition of $\autom_i$ when
processing $v$, where $q^0$ denotes the start state of
$\autom_i$. Finally, there is an edge from $(u,q)$ to $t_i$ if and
only if $q$ is an accepting state of $\autom_i$. Paths in $\grph_i$
correspond to paths in the physical network that satisfy the path
constraints of statement $i$:

\begin{lemma} \label{l:placement}
A sequence of locations $u_1,u_2,\ldots,u_k$ satisfies the constraint
described by regular expression $\bar{a}_i$ if and only if $\grph_i$
contains a path of the form \\$s_i, (u_1,q_1), (u_2,q_2), \ldots,
(u_k,q_k), t_i$ for some state sequence $q_1,\ldots,q_k$.
\end{lemma}
\begin{proof}
The construction of $\grph_i$ ensures that \\$s_i, (u_1,q_1),
(u_2,q_2), \ldots, (u_k,q_k), t_i$ is a path if and only if {\bf (i)}
the sequence $u_1,\ldots,u_k$ represents a path in the physical
network (possibly with vertices of the path repeated more than once
consecutively in the sequence), and {\bf (ii)} the automaton
$\autom_i$ has an accepting computation path for $u_1,\ldots,u_k$ with
state sequence $q^0,q^1,\ldots,q^k$.  The lemma follows from the fact
that a string belongs to the regular language defined by $\bar{a}_i$
if and only if $\autom_i$ has a computation path that accepts that
string.
\end{proof}

Figure~\ref{fig:ilp} illustrates the construction of the graph
$\grph_i$ for a statement with path expression {\tt \small h1 .* dpi
.* nat .* h2}, on a small example network. We assume that deep packet
inspection ({\tt \small dpi}) can be performed at \texttt{\small
h1}, \texttt{\small h2}, or \texttt{\small m1}, whereas network
address translation ({\tt \small nat}) can only be performed
at \texttt{\small m1}. Paths matching the regular expression can be
``lifted'' to paths in $\grph_i$; the thick, red path in the figure
illustrates one such lifting.  Notice that the physical network also
contains other paths such as \texttt{\small h1}, \texttt{\small
s1}, \texttt{\small h2} that do not match the regular
expression. These paths do not lift to any path in $\grph_i$. For
instance, focusing attention on the rows of nodes corresponding to
states $q^2$ and $q^3$ of the \textsc{nfa}, one sees that all edges
between these two rows lead into node ($m_1$,$q^3$). This, in turn,
means that any path that avoids \texttt{\small m1} in the physical
network cannot be lifted to an $s_i$-$t_i$ path in the graph
$\grph_i$.

\begin{figure}[t!]
 \centerline{\begin{tabular}{c@{ }c@{ }c}
  \includegraphics[width=0.25\columnwidth]{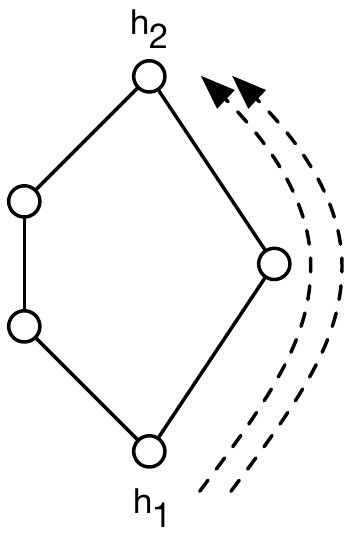} &
  \includegraphics[width=0.25\columnwidth]{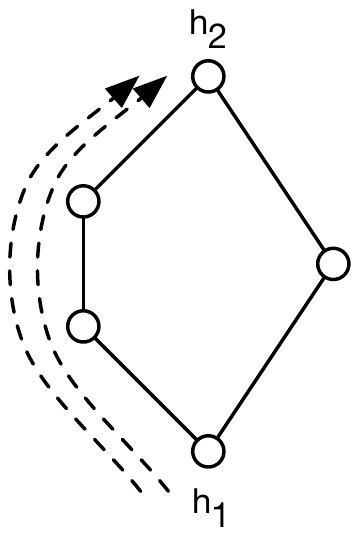} &
  \includegraphics[width=0.25\columnwidth]{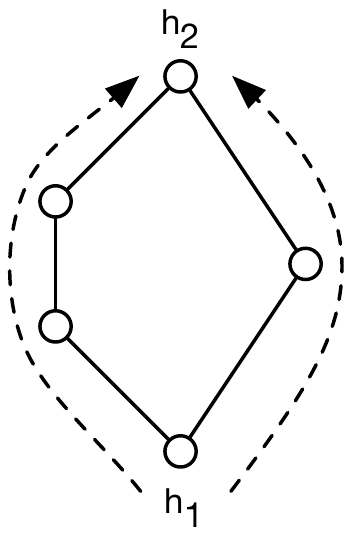} \\
 (a) Shortest-Path & (b) Min-Max Ratio & (c) Min-Max Reserved \\
 \end{tabular}}
  \caption{Path selection heuristics.}
  \label{fig:heuristics}
 \end{figure}

\paragraph*{Path selection}
Next, the compiler determines a satisfying assignment of paths that
respect the bandwidth constraints encoded in the policy. The problem
bears a similarity to the well-known \emph{multi-commodity flow
problem}~\cite{ahuja93}, with two additional types of constraints: (i)
{\em integrality constraints} demand that only one path may be
selected for each statement, and (ii) {\em path constraints} are
specified by regular expressions, as discussed above. To incorporate
path constraints, we formulate the problem in the graph $\grph
= \bigcup_i \grph_i$ described above, rather than in the physical
network itself. Incorporating integrality constraints into
multi-commodity flow problems renders them
\textsc{np}-complete in the worst case, but a number of approaches
have been developed over the years for surmounting this problem,
ranging from approximation
algorithms~\cite{chakrabarti02,chuzhoy12,dinitz99,kleinberg-unsplittable,kolliopoulos01},
to specialized algorithms for topologies such as
expanders~\cite{broder97,frieze98,kleinberg96} and planar
graphs~\cite{okamura-seymour}, to the use of mixed-integer
programming~\cite{barnhart00}. Our current implementation adopts the
latter technique.

Our mixed-integer program (\textsc{mip}) has one $\{0,1\}$-valued
decision variable $x_e$ for each edge $e$ of $\grph$; selecting a
route for each statement corresponds to selecting a path from $s_i$ to
$t_i$ for each $i$ and setting $x_e=1$ on the edges of those paths,
$x_e=0$ on all other edges of $\grph$. These variables are required to
satisfy the flow conservation equations
\begin{align}
   \forall v \in \grph 
   \sum_{e \in \delta^+(v)} x_e 
       -  
   \sum_{e \in \delta^-(v)} x_e 
       = 
   \begin{cases}
        1 & \mbox{\small if $v = s_i$} \\
        -1 & \mbox{\small if $v = t_i$} \\
        0 & \mbox{\small otherwise}
   \end{cases}
\label{eq:flowcons}
\end{align}
where $\delta^+(v), \, \delta^-(v)$ denote the sets of edges exiting
and entering $v$, respectively. For bookkeeping purposes
the \textsc{mip} also has real-valued variables $\nr_{uv}$ for each
physical network link $(u,v)$, representing what fraction of the
link's capacity is reserved for statements whose assigned path
traverses $(u,v)$. Finally, there are variables $\nrmax$ and $\unrmax$
representing the maximum fraction of any link's capacity devoted to
reserved bandwidth, and the maximum net amount of reserved bandwidth
on any link, respectively. The equations and inequalities pertaining
to these additional variables can be written as follows.  For any
statement $i$, let $\rmini$ denote the minimum amount of bandwidth
guaranteed in the rate clause of statement $i$. ($\rmini=0$ if the
statement contains no bandwidth guarantee.)  For any physical link
$(u,v)$, let $c_{uv}$ denote its capacity and let $E_i(u,v)$ denote
the set of all edges of the form $((u,q),(v,q'))$ or $((v,q),(u,q'))$
in $\grph_i$.
\begin{align}
  \forall (u,v) \; \; 
  \nr_{uv} c_{uv}
     &=
  \sum_i \sum_{e \in E_i(u,v)} \rmini x_e 
      \label{eq:nr} \\
  \forall (u,v) \; \; 
  \nrmax & \geq \nr_{uv} 
      \label{eq:nrmax} \\
  \forall (u,v) \; \;
  \unrmax & \geq \nr_{uv} c_{uv} 
      \label{eq:unrmax} \\
  \nrmax & \leq 1 
      \label{eq:capacity}
\end{align}
Constraint~\ref{eq:nr} defines $\nr_{uv}$ to be the fraction of
capacity on link $(u,v)$ reserved for bandwidth guarantees.
Constraints~\ref{eq:nrmax} and~\ref{eq:unrmax} ensure that $\nrmax$
(respectively, $\unrmax$) is at least the maximum fraction of capacity
reserved on any link (respectively, the maximum net amount of
bandwidth reserved on any link). Constraint~\ref{eq:capacity} ensures
that the route assignment will not exceed the capacity of any link, by
asserting that the fraction of capacity reserved on any link cannot
exceed $1$.

\paragraph*{Path selection heuristics}
In general, there may be multiple assignments that satisfy the path
and bandwidth constraints. To indicate the preferred assignment,
programmers can invoke Merlin with one of three optimization criteria:
\begin{itemize} 
\item {\em Weighted shortest path:} minimizes the total
number of hops in assigned paths, weighted by bandwidth
guarantees: \(\min \sum_i \sum_{u \neq v} \sum_{e \in E_i(u,v)} \rmini
x_e\).  This criterion is appropriate when the goal is to minimize
latency as longer paths tend to experience increased latency.
\item {\em Min-max ratio:} minimizes the maximum fraction of capacity
 reserved on any link (i.e., $\nrmax$). This criterion is appropriate
when the goal is to balance load across the network links.
\item {\em Min-max reserved:} minimizes the maximum amount of bandwidth 
reserved on any link (i.e., $\unrmax$).  This criterion is appropriate
  when the goal is to guard against failures, since it limits the
  maximum amount of traffic that may disrupted by a single link failure.
\end{itemize}
The differences between these heuristics are illustrated in
Figure~\ref{fig:heuristics} which depicts a simple network with
hosts \texttt{\small h1} and \texttt{\small h2} connected by a pair of disjoint
paths. The left path comprises three edges of capacity $400$MB/s. The
right path comprises two edges of capacity $100$MB/s. The figure shows
the paths selected for two statements each requesting a bandwidth
guarantee of $50$MB/s. Depending on the heuristic, the \textsc{mip}
solver will either select two-hop paths (weighted shortest path),
reserve no more than $25$\% of capacity on any link (min-max ratio),
or reserve no more than $50$MB/s on any link (min-max reserved).

\subsection{Provisioning for Best-Effort Rates}

For traffic requiring only best-effort rates, Merlin does not need to
solve a constraint problem. Instead, the compiler only needs to
compute sink-trees that obey the path constraints expressed in the
policy. A sink-tree for a particular network node forwards traffic
from elsewhere on the network to that node. Merlin does this by
computing the cross product of the regular expression NFA and the
network topology representation, as just described, and then
performing a breath-first search over the resulting graph. To further
improve scalability, the compiler uses a small optimization: it works
on a topology that only includes switches, and computes a sink tree
for each egress switch. The compiler adds instructions to forward
traffic from the egress switches to the hosts during code
generation. This allows the BFS to be computed in $O(|V||E|)$, where
$|V|$ is the number of switches rather than the number of hosts.

\subsection{Code Generation}
\label{sec:enforcement}

Next, the Merlin compiler generates code to enforce the policy using
the devices available in the network. The actual code is determined
both by the requested functionality and the type of target device. For
basic forwarding and bandwidth guarantees, Merlin generates
instructions for OpenFlow switches and controllers~\cite{mckeown08}
that install forwarding rules and configure port queues. For packet
transformations such as deep packet inspection, network address
translation, and intrusion detection, Merlin generates scripts to
install and configure middleboxes, such as
Click~\cite{kohler00}. Traffic filtering and rate limiting are
implemented by generating calls to the standard Linux
utilities~\texttt{\small iptables}. Of course, other approaches are
possible. For example, Merlin could implement packet transformations
by using Puppet~\cite{puppet} to provision virtual machines that
implement those transformations.

Because Merlin controls forwarding paths but also supports middleboxes
that may modify headers (such as \textsc{nat} boxes), the compiler
needs to use a forwarding mechanism that is robust to changes in
packet headers. Our current implementation uses \textsc{vlan} tags to
encode paths to destination switches, one tag per sink tree. All
packets destined for that tree's sink are tagged with a tag
when they enter the network. Subsequent switches simply examine the
tag to determine the next hop. At the egress switch, the tag is
stripped off and a unique host identifier (e.g., the \textsc{mac}
address) is used to forward traffic to the appropriate host. This
approach is similar to the technique used in
FlowTags~\cite{fayazbakhsh14}.

Merlin can provide greater flexibility and expressiveness by directly
generating packet-processing code, which can be executed by an
interpreter running on end hosts or on middleboxes. We have also built
a prototype that runs as a Linux kernel module and uses the
\texttt{\small netfilter} callback functions to access packets on the network
stack. The interpreter accepts and enforces programs that can filter
or rate limit traffic using a richer set of predicates than those
offered by \texttt{\small iptables}. It is designed to have minimal
dependencies on operating system services in order to make it portable
across different systems. The current implementation requires only
about a dozen system calls to be exported from the operating system to
the interpreter. In on-going work, we are exploring additional
functionality, with the goal or providing a general runtime as the
target for the Merlin complier. However, using end hosts assumes a
trusted deployment in which all host machines are under administrative
control. An interesting, but orthogonal, problem is to deploy Merlin
in an untrusted environment.  Several techniques have been proposed to
verify that an untrusted machine is running certain software. Notable
examples include proof carrying code~\cite{necula97}, and
\textsc{tpm}-based attestations~\cite{dixon11,sirer11}.

Merlin can be instantiated with a variety of backends to capitalize on
the capabilities of the available devices in the network. Although the
expressiveness of policies is bounded by the capabilities of the
devices, Merlin provides a unified interface for programming them.

\section{Dynamic Adaptation}
\label{sec:transformations}

The Merlin compiler described in the preceding section translates
policies into static configurations. Of course, these static
configurations may under-utilize resources, depending on how traffic
demands evolve over time. Moreover, in a shared environment, network
tenants may wish to customize global policies to suit their own
needs---e.g., to add additional security constraints. 

To allow for the dynamic modification of policies, Merlin uses small
run-time components called \emph{negotiators}. Negotiators are policy
transformers and verifiers---they allow policies to be delegated to
tenants for modification and they provide a mechanism for verifying
that modifications made by tenants do not lead to violations of the
original global policy. Negotiators depend critically on Merlin's
language-based approach. The same abstractions that allow policies to
be mapped to constraint problems (i.e., predicates, regular
expressions, and explicit bandwidth reservations), make it easy to
support \textit{verifiable} policy transformations.

Negotiators are distributed throughout the network in a tree, forming
a hierarchical overlay over network elements. Each negotiator is
responsible for the network elements in the subtree for which it is
the root. Parent negotiators impose policies on their
children. Children may refine their own policies, as long as the
refinement implies the parent policy. Likewise, siblings may
renegotiate resource assignments cooperatively, as long as they do not
violate parent policies. Negotiators communicate amongst themselves to
dynamically adjust bandwidth allocations to fit particular deployments
and traffic demands.

\subsection{Transformations}
\label{sec:delegation}

With negotiators, tenants can transform global network policies by
\emph{refining} the delegated policies to suit their own demands. 
Tenants may modify policies in three ways: (i) policies may be refined
with respect to packet classification; (ii) forwarding paths may be
further constrained; and (iii) bandwidth allocations may be revised.

\paragraph*{Refining policies}
Merlin policies classify packets into sets using predicates that
combine matches on header fields using logical operators. These sets
can be refined by introducing additional constraints to the original
predicate. For example, a predicate for matching all
\textsc{tcp} traffic:
\begin{alltt}\small
   ip.proto = tcp
\end{alltt}
can be partitioned into ones that match \textsc{http} traffic and all
other traffic:
\begin{alltt}\small
   ip.proto = tcp and tcp.dst = 80
   ip.proto = tcp and tcp.dst != 80
\end{alltt}
\noindent
The partitioning must be total---all packets identified by the
original policy must be identified by the set of new policies.

\paragraph*{Constraining paths}
Merlin programmers declare path constraints using regular expressions
that match sequences of network locations or packet
transformations. Tenants can refine a policy by adding addition
constraints to the regular expression. For example, an expression that
says all packets must go through a traffic logger (\textsc{log})
function:
\begin{alltt}\small
   .* log .*
\end{alltt}
can be modified to say that the traffic must additionally pass through
a \textsc{dpi} function:
\begin{alltt}\small
   .* log .* dpi .*
\end{alltt}

\paragraph*{Re-allocating bandwidth}
Merlin's limits (\texttt{\small max}) and guarantees (\texttt{\small min}) constrain
allocations of network bandwidth. After a policy has been refined,
these constraints can be redistributed to improve utilization. The
requirement for a valid transformation is that the sum of the new
allocations must not exceed the original allocation.

\paragraph*{Example} 
As an example that illustrates the use of all three transformations,
consider the following policy, which caps all traffic between two
hosts at 100MB/s:

\begin{alltt}\small
   [x : (ip.src = 192.168.1.1 and
         ip.dst = 192.168.1.2) -> .*],
   max(x, 100MB/s)
\end{alltt}

\noindent
This policy could be modified as follows:

\begin{alltt}\small
   [x : (ip.src = 192.168.1.1 and
         ip.dst = 192.168.1.2 and
         tcp.dst = 80) -> .* log .*],
   [y : (ip.src = 192.168.1.1 and
         ip.dst = 192.168.1.2 and
         tcp.dst = 22) -> .* ],
   [z : (ip.src = 192.168.1.1 and
         ip.dst = 192.168.1.2 and
         !(tcpDst=22|tcpDst=80)) -> .* dpi .*],
   max(x, 50MB/s)
   and max(y, 25MB/s)
   and max(z, 25MB/s)
\end{alltt}

\noindent
It gives $50$MB/s to \textsc{http} traffic with the constrain that it
passes through a \texttt{\small log} that monitors all web requests, $25$MB/s
to \textsc{ssh} traffic, and $25$MB/s to the remaining traffic, which
must flow through a \texttt{\small dpi} box.

\subsection{Verification}

Allowing tenants to make arbitrary modifications to policies would not
be safe. For example, a tenant could lift restrictions on forwarding
paths, eliminate transformations, or allocate more bandwidth to their
own traffic---all violations of the global policy set down by the
administrator. Fortunately, Merlin negotiators can leverage the policy
language representation to check policy inclusion, which can be used
to establish the correctness of policy transformations implemented by
untrusted tenants.

Intuitively, a valid refinement of a policy is one that makes it only
more restrictive. To verify that a policy modified by a tenant is a
valid refinement of the original, the negotiator simply has to check
that for every statement in the original policy, the set of paths
allowed for matching packets in the refined policy is included in the
set of paths in the original, and the bandwidth constraints in
the refined policy imply the bandwidth constraints in the
original. These conditions can be decided using a simple algorithm
that performs a pair-wise comparison of all statements in the original
and modified policies, (i) checking for language
inclusion~\cite{hopcroft79} between the regular expressions in
statements with overlapping predicates, and (ii) checking that the sum
of the bandwidth constraints in all overlapping predicates implies the
original constraint.

\subsection{Adaptation}

Bandwidth re-allocation does not require recompilation of the global
policy, and can thus happen quite rapidly. As a proof-of-concept, we
implemented negotiators that can provide both min-max fair sharing and
additive-increase, multiplicative decrease allocation schemes. These
negotiators allow traffic policies to change with dynamic workloads,
while still obeying the overall static global policies. Changes in
path constraints require global recompilation and updating forwarding
rules on the switches, so they incur a greater overhead. However, we
believe these changes are likely to occur less frequently than changes
to bandwidth allocations.

 \section{Implementation}
\label{sec:implementation}

We have implemented a full working prototype of the Merlin system in
OCaml and C. Our implementation uses the Gurobi
Optimizer~\cite{gurobi} to solve constraints, the Frenetic
controller~\cite{frenetic-ocaml13} to install forwarding rules on
OpenFlow switches, the Click router~\cite{kohler00} to manage software
middleboxes, and the \texttt{\small ipfilters} and \texttt{\small tc} utilities on
Linux end hosts. Note that the design of Merlin does not depend on
these specific systems. It would be easy to instantiate our design
with other systems, and our implementation provides a clean interface
for incorporating additional backends.

Our implementation of Merlin negotiator and verification mechanisms
leverages standard algorithms for transforming and analyzing
predicates and regular expressions. To delegate a policy, Merlin
simply intersects the predicates and regular expressions in each
statement the original policy to project out the policy for the
sub-network. To verify implications between policies, Merlin uses the
Z3 SMT solver~\cite{demoura08} to check predicate disjointness, and
the Dprle library~\cite{dprle} to check inclusions between regular
expressions.

\begin{figure}[t]
\centerline{\includegraphics[trim=7mm 7mm 7mm 2cm, clip=true,width=\columnwidth]{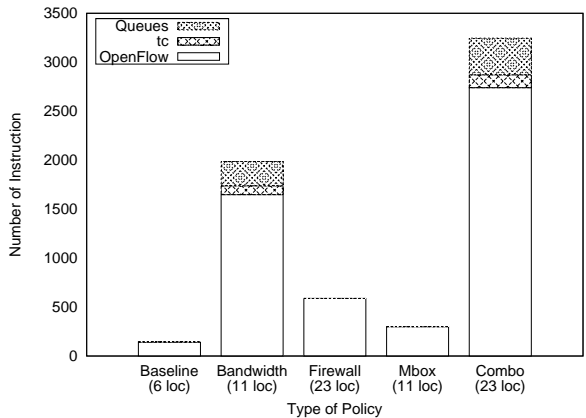}}
\vspace{-7mm}
\caption{Merlin expressiveness, measured using policies for the Stanford campus network topology.}
\label{fig:stanford}
\end{figure}

\section{Evaluation}
\label{sec:evaluation}

To evaluate Merlin, we investigated three main issues: (i) the
expressiveness of the Merlin policy language, (ii) the ability of
Merlin to improve end-to-end performance for applications, and (iii)
the scalability of the compiler and negotiator components with respect
to network and policy size.
We used two testbeds in our evaluation. Most experiments were run on a
cluster of Dell r720 PowerEdge servers with two $8$-core 2.7GHz Intel
Xeon processors, 32GB RAM, and four 1GB NICs. The Ring Paxos
experiment (\S\ref{sec:management}) was conducted on a cluster of
eight HP SE1102 servers equipped with two quad-core Intel Xeon L5420
processors running at 2.5 GHz, with 8 GB of RAM and two 1GB NICs. Both
clusters used a Pica8 Pronto 3290 switch to connect the machines. To
test the scalability we ran the compiler and negotiator frameworks on
various topologies and policies.

Overall, our experimental evaluation shows that Merlin can quickly
provision and configure real-world datacenter and enterprise networks,
that Merlin can be used to obtain better performance for big-data
processing applications and replication systems, and that Merlin
allows administrators to succinctly express network policies.

\subsection{Expressiveness.}

To explore the expressiveness of the Merlin policy languages, we built
several network policies for the $16$-node Stanford core campus
network topology~\cite{stanford}. We created $24$ subnets and then
implemented a series of policies in Merlin, and compared the sizes
of the Merlin source policies and the outputs generated by the
compiler. This comparison measures the degree to which Merlin is able
to abstract away from hardware-level details and provide effective
constructs for managing a real-world network.

The Merlin policies we implemented are as follows:

\begin{enumerate}\setlength{\itemsep}{0pt}

\item \emph{All-pairs connectivity}. This policy implements basic 
  forwarding between all pairs of hosts and provides a baseline
  measurement of the number of low-level instructions that would be
  needed in almost any non-trivial application. The Merlin policy is
  only $6$ lines long and compiles to $145$ OpenFlow rules.
\item \emph{Bandwidth caps and guarantee}. This policy augments the
  basic connectivity by providing $10$\% of traffic classes a
  bandwidth guarantee of $1$Mbps and a cap of $1$Gbps. Such a
  guarantee would be useful, for example, to prioritize emergency
  messages sent to students. This policy required $11$ lines of Merlin
  code, but generates over $1600$ OpenFlow rules, $90$ \textsc{tc}
  rules and $248$ queue configurations. The number of OpenFlow rules
  increased dramatically due to the presence of the bandwidth
  guarantees which required provisioning separate forwarding paths for
  a large collection of traffic classes.
\item \emph{Firewall}. This policy assumes the presence of a
  middlebox that filters incoming web traffic connected to the network
  ingress switches. The baseline policy is altered to forward all
  packets matching a particular pattern (e.g., \texttt{\small tcp.dst = 80})
  through the middlebox. This policy requires $23$ lines of Merlin
  code, but generates over $500$ OpenFlow instructions.
\item \emph{Monitoring middlebox}. This policy
  attaches middleboxes to two switches and partitions the hosts into
  two sets of roughly equal size. Hosts connected to switches in the
  same set may send traffic to each other directly, but traffic
  flowing between sets must be passed through a middlebox. This policy
  is useful for filtering traffic from untrusted sources, such as
  student dorms. This policy required $11$ lines of Merlin code but
  generates $300$ OpenFlow rules, roughly double the baseline number.
\item \emph{Combination}. This policy augments the basic
  connectivity with a filter for web traffic, a bandwidth guarantee
  for certain traffic classes and an inspection policy for a certain
  class of hosts. This policy requires $23$ lines of Merlin code, but
  generates over $3000$ low-level instructions.
\end{enumerate}

The results of this experiment are depicted in
Figure~\ref{fig:stanford}. Overall, it shows that using Merlin
significantly reduces the effort, in terms of lines of code, required
to provision and configure network devices for a variety of real-world
management tasks.

\subsection{Application Performance}
\label{sec:management}

Our second experiment measured Merlin's ability to improve end-to-end
performance for real-world applications.

\paragraph*{Hadoop}
Hadoop is a popular open-source MapReduce~\cite{dean04}
implementation, and is widely-used for data analytics. A Hadoop
computation proceeds in three stages: the system (i) applies
a \emph{map} operator to each data item to produce a large set of
key-value pairs; (ii) \emph{shuffles} all data with a given key to a
single node; and (iii) applies the \emph{reduce} operator to values
with the same key. The many-to-many communication pattern used in the
shuffle phase often results in heavy network load, making Hadoop jobs
especially sensitive to background traffic. In practice, this
background traffic can come from a variety of sources. For example,
some applications use \textsc{udp}-based gossip protocols to update
state, such as system monitoring tools~\cite{astrolabe,gems}, network
overlay management~\cite{tman}, and even distributed storage
systems~\cite{astrolabe,dynamo}. A sensible network policy would be to
provide guaranteed bandwidth to Hadoop jobs so that they finish
expediently, and give the \textsc{udp} traffic only best-effort
guarantees.

With Merlin, we implemented this policy using just three policy
statements. To show the impact of the policy, we ran a Hadoop job that
sorts $10$GB of data, and measured the time to complete it on a
cluster of four servers, under three different configurations:
\begin{enumerate}\setlength{\itemsep}{0pt}
\item \emph{Baseline.} Hadoop had exclusive access to the network. 
\item \emph{Interference.} we used
the \texttt{\small iperf} tool to inject \textsc{udp} packets, simulating
background traffic.
\item \emph{Guarantees.} we again injected background traffic, 
but guaranteed $90$ percent of the capacity for Hadoop.
\end{enumerate}

The measurements demonstrate the expected results. With exclusive
network access, the Hadoop job finished in $466$ seconds. With
background traffic causing network congestion, the job finished in
$558$ seconds, a roughly $20$\% slow down. With the Merlin policy
providing bandwidth guarantees, the job finished in $500$ seconds,
corresponding to the $90$\% allocation of bandwidth. 

\begin{figure}[t]
\centerline{\begin{tabular}{c@{}c}
  \includegraphics[width=0.5\columnwidth]{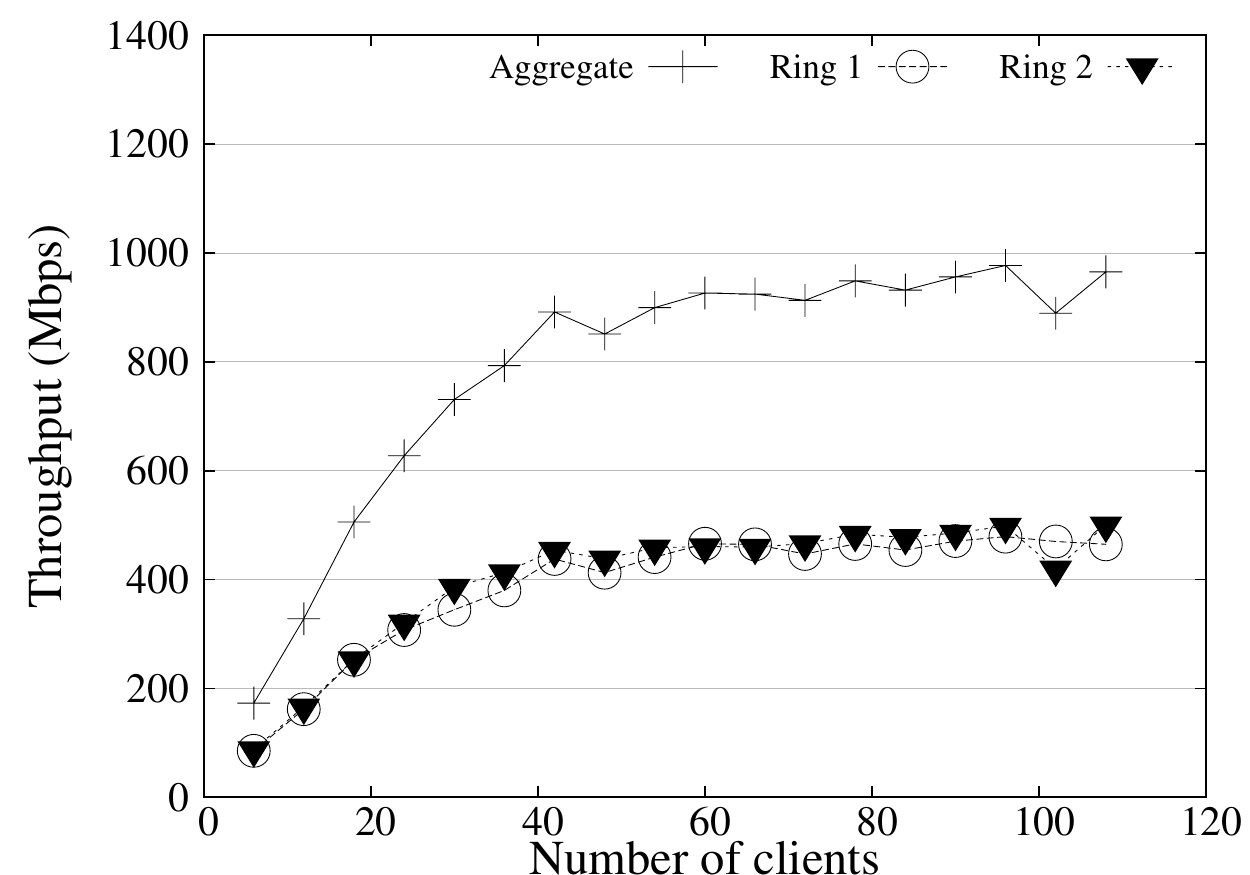} &
  \includegraphics[width=0.5\columnwidth]{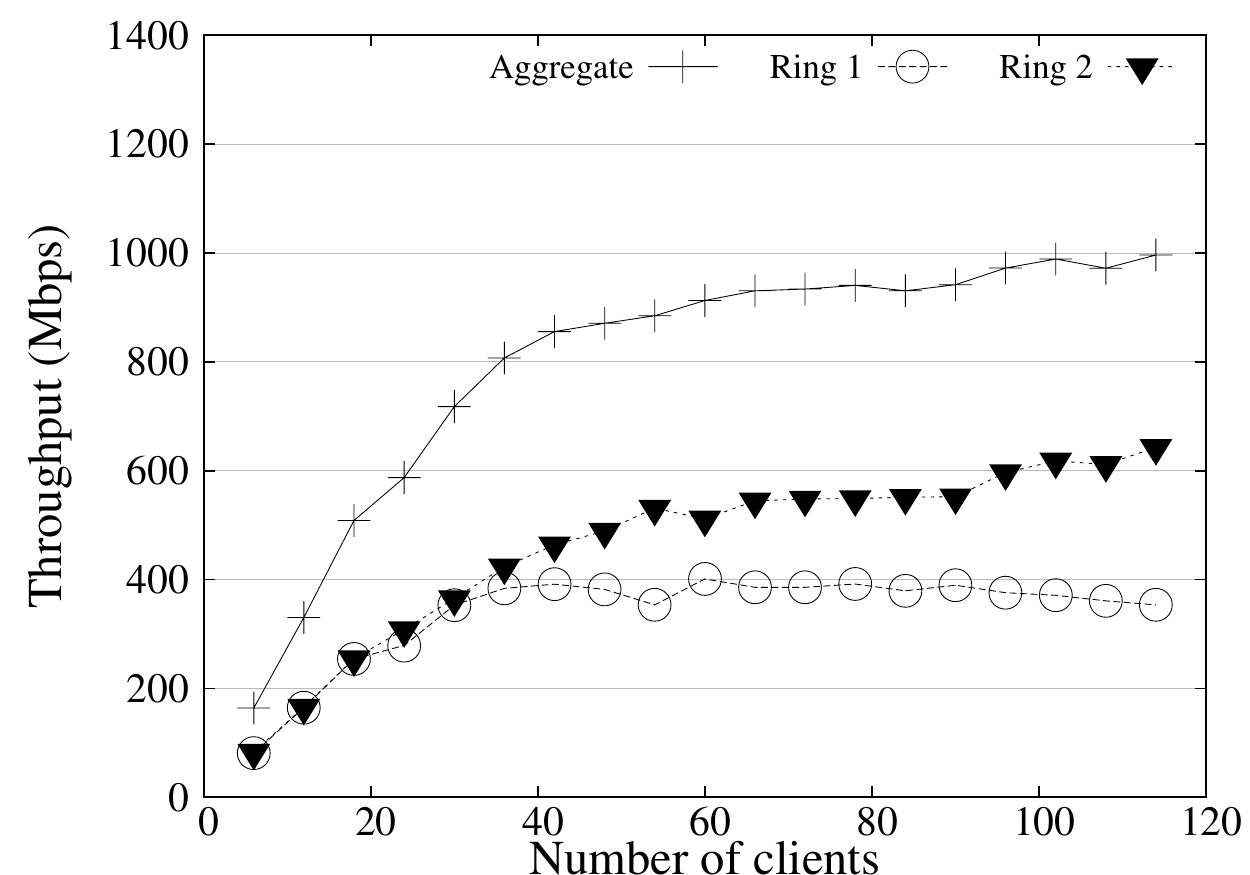} \\
(a) & (b) \\
\end{tabular}}
\vspace{-4mm}
\caption{Ring-Paxos (a) without and (b) with Merlin.}
\label{fig:ringpaxos}
\end{figure}

\paragraph*{Ring-Paxos}
State-machine replication (SMR) is a fundamental approach to designing
fault-tolerant services~\cite{lamport78,schneider90} at the core of
many current systems (e.g., Google’s Chubby~\cite{burrows06},
Scatter~\cite{glendenning11}, Spanner~\cite{corbett12}). State machine
replication provides clients with the abstraction of a highly
available service by replicating the servers and regulating how
commands are propagated to and executed by the replicas: (i) every
nonfaulty replica must receive all commands in the same order; and
(ii) the execution of commands must be deterministic.

Because ordering commands in a distributed setting is a non-negligible
operation, the performance of a replicated service is often determined
by the number of commands that can be ordered per time unit.  To
achieve high performance, the service state can be partitioned and
each partition replicated individually (e.g., by separating data from
meta-data), but the partitions will compete for shared resources
(e.g., common nodes and network links).

We assessed the performance of a key-value store service replicated
with state-machine replication.  Commands are ordered using an
open-source implementation of Ring Paxos~\cite{MPSP10}, a highly
efficient implementation of the Paxos protocol~\cite{lamport98}.  We deployed two instances
of the service, each one using four processes.  One process in each
service is co-located on the same machine and all other processes run
on different machines.  Clients are distributed across six different
machines and submit their requests to one of the services and receive
responses from the replicas.

Figure~\ref{fig:ringpaxos} (a) depicts the throughput of the two
services; the aggregate throughput shows the accumulated performance
of the two services.  Since both services compete for resources on the
common machine, each service has a similar share of the network, the
bottlenecked resource at the common machine.  In
Figure~\ref{fig:ringpaxos} (b), we provide a bandwidth guarantee for
Service 2.  Note that this guarantee does not come at the expense of
utilization. If Service 2 stops sending traffic, Service 1 is free to use
the available bandwidth.

\paragraph*{Summary}
Overall, these experiments show that Merlin policies can concisely
express real-world policies, and that the Merlin system is able to
generate code that achieves the desired outcomes for applications on
real hardware.

\begin{figure}[t]
\centerline{\includegraphics[trim=0mm 1cm 0mm 6.5cm, clip=true, width=0.4\textwidth]{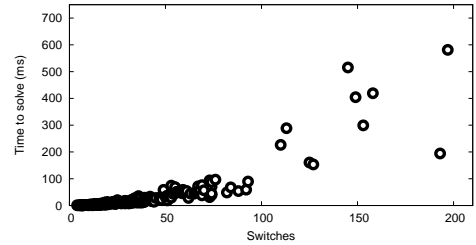}}
\vspace{-7mm}
 \caption{Compilation times for Internet Topology Zoo.}
 \label{fig:zoo}
\end{figure}

\begin{figure*}
\centerline{\begin{tabular}[h]{|c|c|c|c|c|c|}
\hline
\textbf{Traffic Classes} & 
\textbf{Hosts} & 
\textbf{Switches} & 
\textbf{LP construction (ms)} & 
\textbf{LP solution (ms)} & 
\textbf{Rateless solution (ms)}\\
\hline
870 & 30 & 45 & 25 & 22 & 33 \\
8010 & 90 & 80 & 214 & 160 & 36 \\
28730 & 170 & 125 & 364 & 252 & 106 \\
39800 & 200 & 125 & 1465 & 1485 & 91 \\
95790 & 310 & 180 & 13287 & 248779 & 222 \\
136530 & 370 & 180 & 27646 & 1200912 & 215 \\
159600 & 400 & 180 & 29701 & 1351865 & 212 \\
229920 & 480 & 245 & 86678 & 10476008 & 451\\
\hline
\end{tabular}
}
\caption{Number of traffic classes,  
  topology sizes, and solution times for fat tree topologies with 5\%
  of the traffic classes with guaranteed bandwidth.}
\label{table:classes}
\end{figure*}

\begin{figure*}[t]
\centerline{\begin{tabular}{c@{}c@{}c@{}c}
   \includegraphics[trim=0mm 4cm 0mm 0mm, clip=true, width=0.25\textwidth]{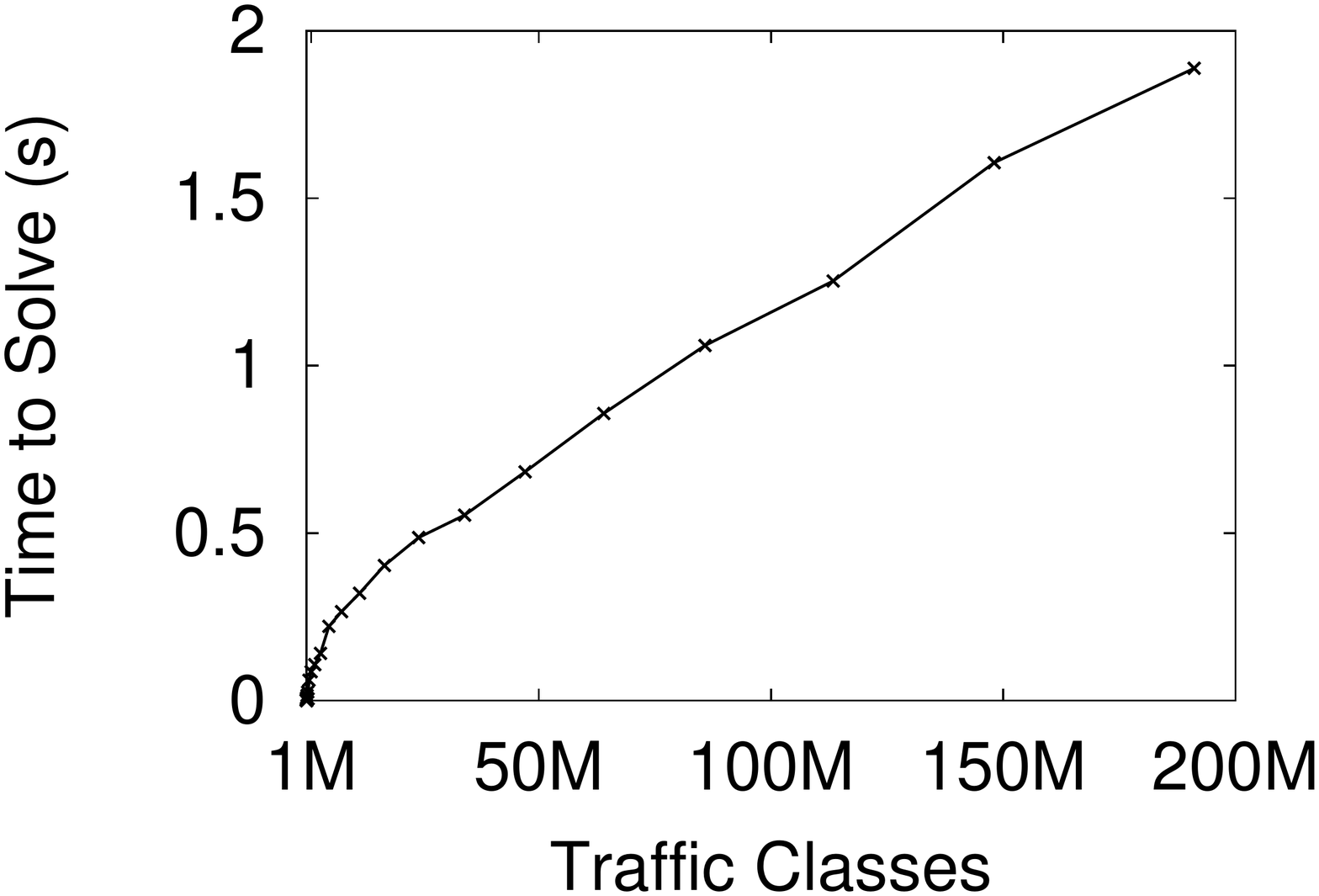} &
   \includegraphics[trim=0mm 4cm 0mm 0mm, clip=true, width=0.25\textwidth]{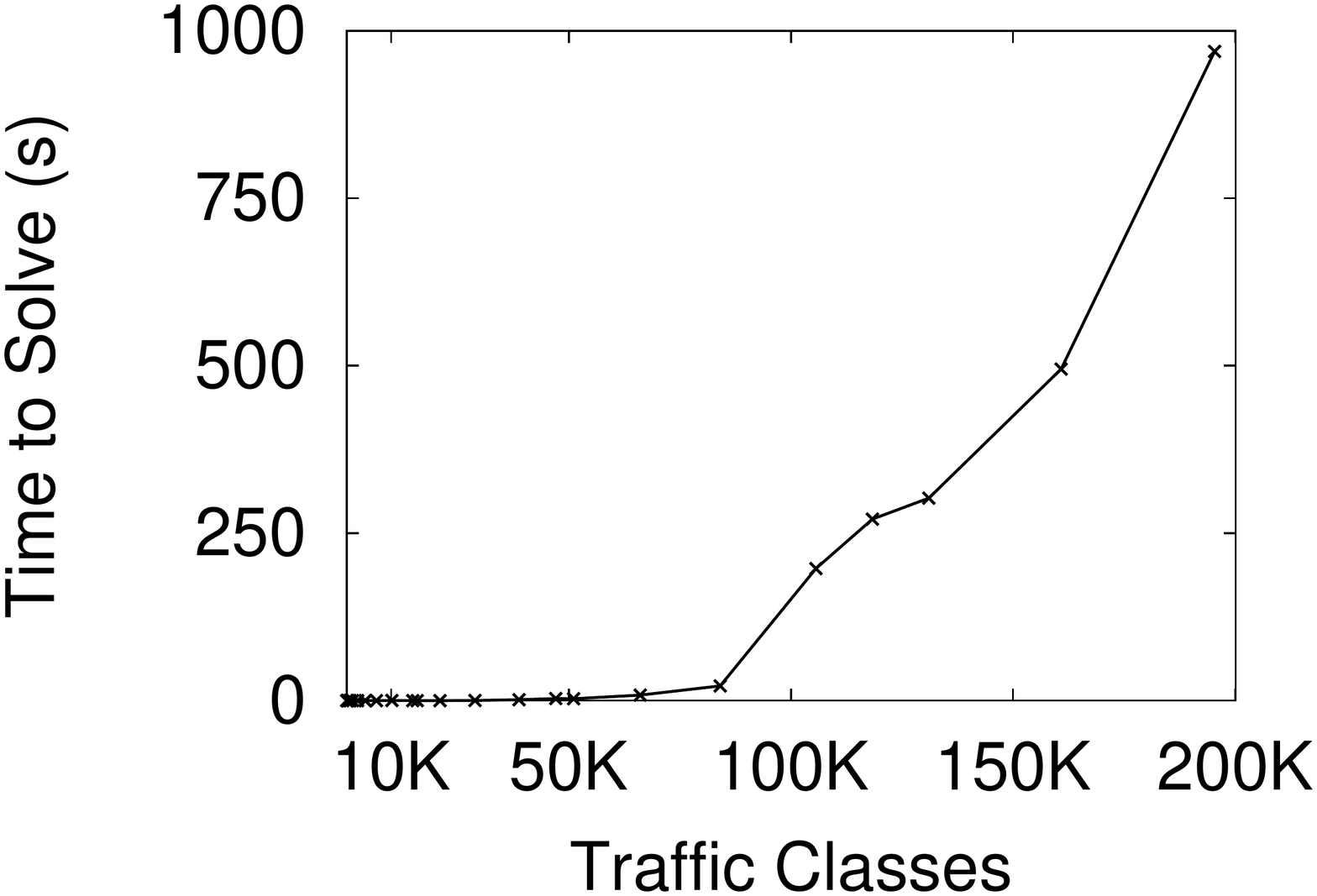} &
   \includegraphics[trim=0mm 4cm 0mm 0mm, clip=true, width=0.25\textwidth]{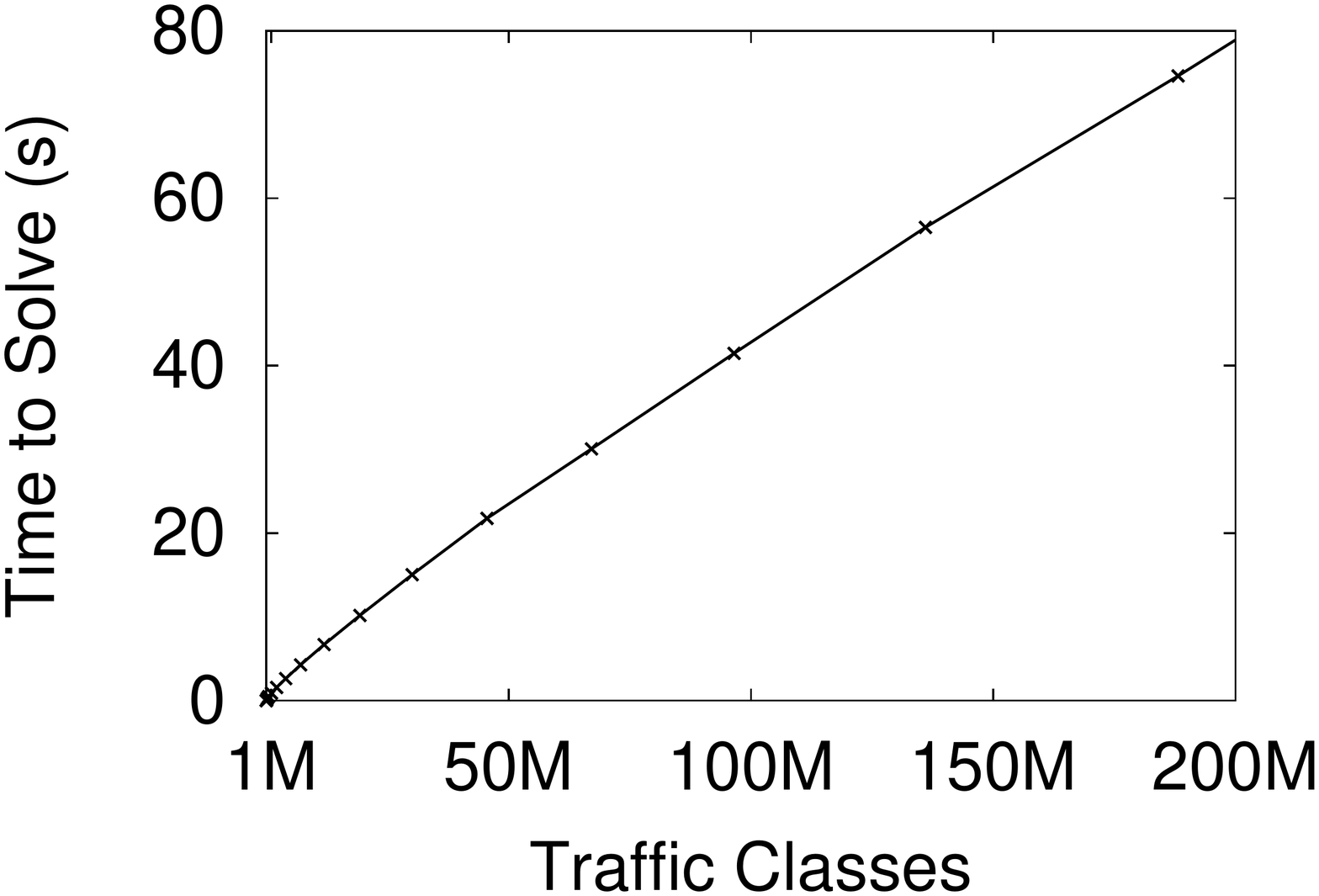} &
   \includegraphics[trim=0mm 4cm 0mm 0mm, clip=true, width=0.25\textwidth]{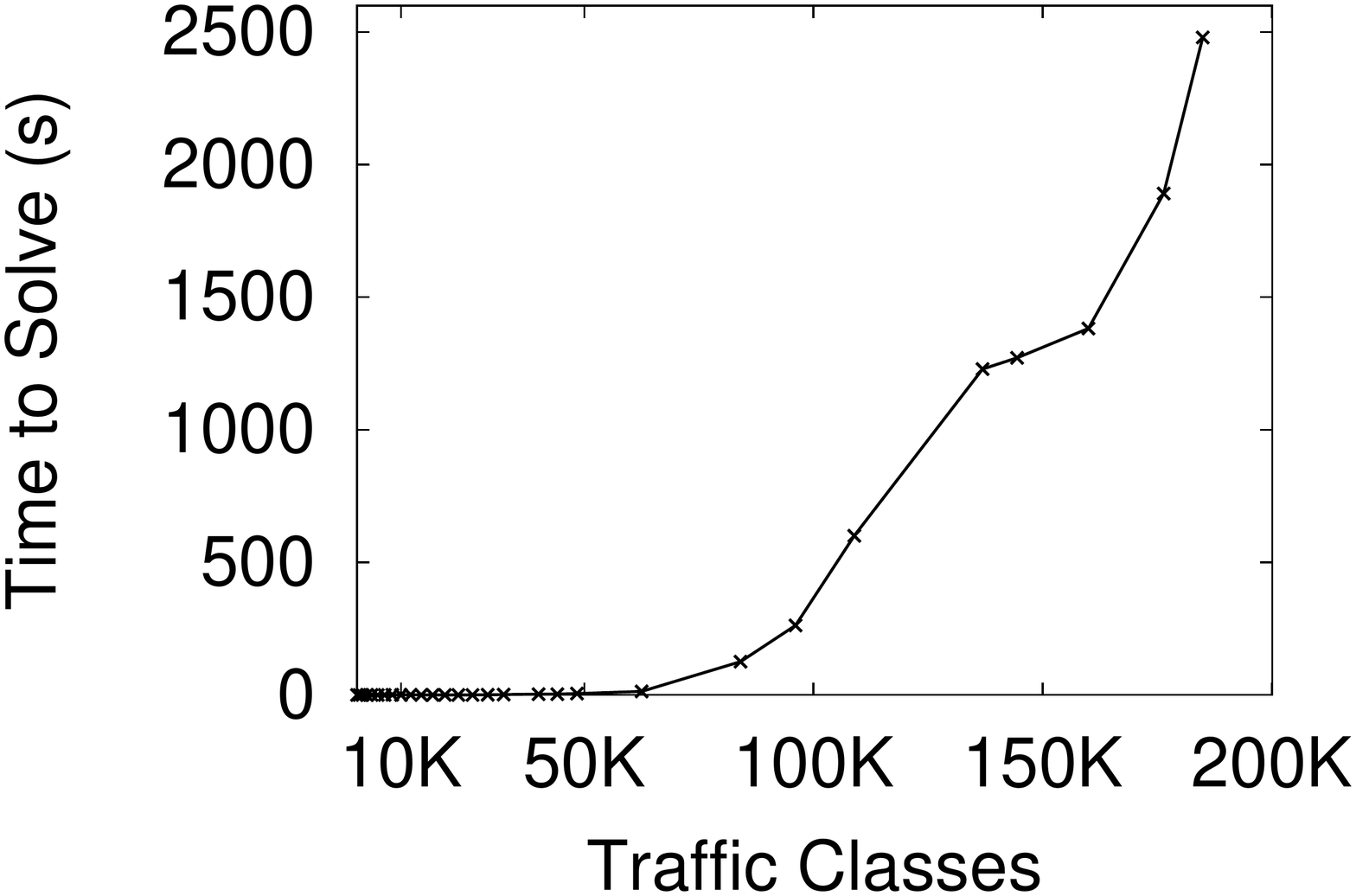} \\
(a) & (b) & (c) & (d) \\
\end{tabular}}
\caption{Compilation times for an increasing number of traffic classes for (a) all pairs
  connectivity on a balanced tree, (b) 5\% of the traffic with guaranteed priority
   on a balanced tree, (c) all pairs connectivity on a fat tree, (d) 5\% of the traffic with guaranteed priority
   on a fat tree.}
\label{fig:scalability}
\end{figure*}

\subsection{Compilation and Verification}

The scalability of the Merlin compiler and verification framework
depend on both the size of the network topology and the number of
traffic classes. Our third experiment evaluates the scalability of
Merlin under a variety of scenarios. 

\paragraph*{Compiler}
We first measured the compilation time needed by Merlin to provide
pair-wise connectivity between all hosts in a topology. This task,
which could be computed offline, has been used to evaluate other
systems, including VMware's NSX, which reports approximately $30$
minutes to achieve $100$\% connectivity from a cold
boot~\cite{koponen14}. We used the Internet Topology Zoo~\cite{zoo}
dataset, which contains $262$ topologies that represent a large
diversity of network structures. The topologies have average size of 40
switches, with a standard deviation of 30 switches. The Merlin
compiler takes less than 50ms to provide connectivity for the majority
of topologies, and less than 600ms to provide connectivity for all but one of the
topologies. Figure~\ref{fig:zoo} shows the results for $261$ of the
topologies. To improve the readability of the graph, we elided the
largest topology, which has $754$ switches and took Merlin $4$ seconds
to compile.

To explore the scalability with bandwidth guarantees, we measured the
compilation time on two types of tree topologies, balanced and fat trees, for an increasing
number of traffic classes. Each traffic class represents a
unidirectional stream going from one host at the edge of the network
to another. Thus, the total number of classes correspond to the number
of point-to-point traffic flows. Table \ref{table:classes} shows a
sample of topology sizes and solution times for various traffic
classes for fat tree topologies. For both types of topologies,
we took two sets of measurements: the
time to provide pair-wise connectivity with no guarantees, and the
time to provide connectivity when 5\% of the traffic classes receive
guarantees. Figure~\ref{fig:scalability} shows the results. As
expected, providing bandwidth guarantees adds overhead to the
compilation time.  For the worst case scenario that we measured, on a
network with $400,000$ total traffic classes, with $20,000$ of those
classes receiving bandwidth guarantees, Merlin took around 41 minutes
to find a solution. To put that number in perspective,
B4~\cite{jain13} only distinguishes $13$ traffic classes. Merlin finds
solutions for $100$ traffic classes with guarantees in a network with
$125$ switches in less than $5$ seconds.

These experiments show that Merlin can provide connectivity for large
networks quickly and our mixed-integer programming approach used for
guaranteeing bandwidth scales to large networks with reasonable
overhead.

\begin{figure*}[t]
\centerline{\begin{tabular}{c@{}c@{}c}
  \includegraphics[width=0.3\textwidth]{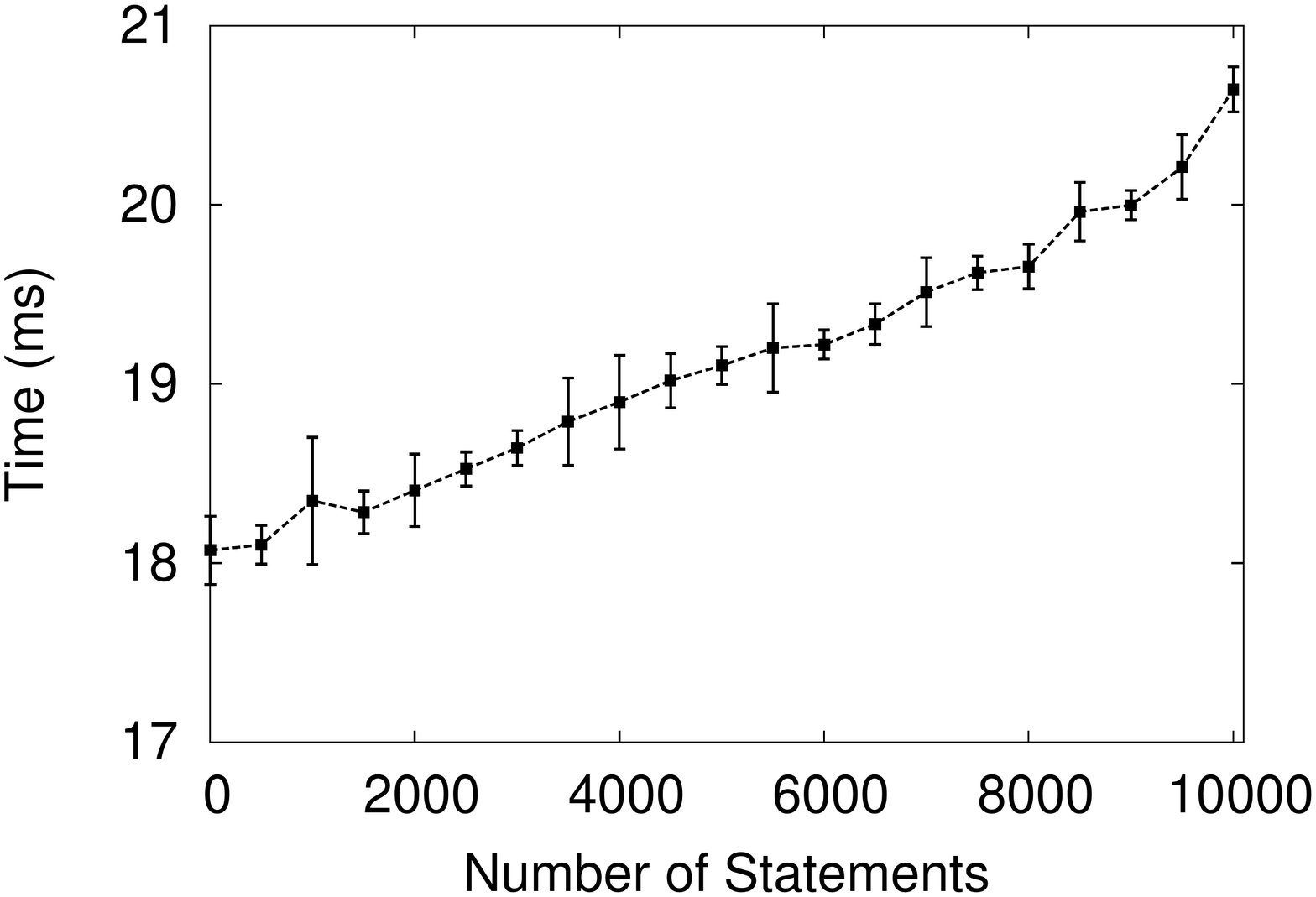} &
  \includegraphics[width=0.3\textwidth]{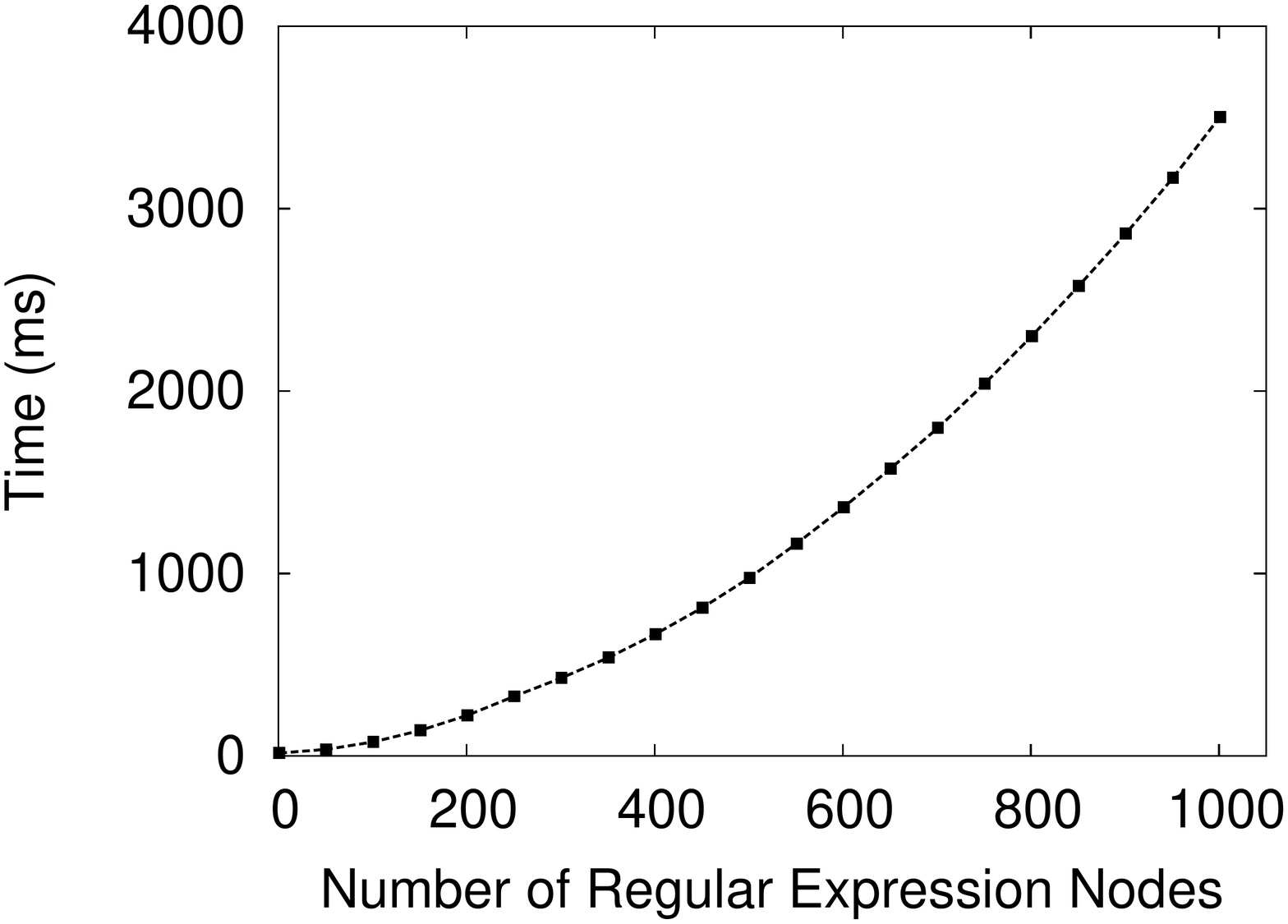} &
  \includegraphics[width=0.3\textwidth]{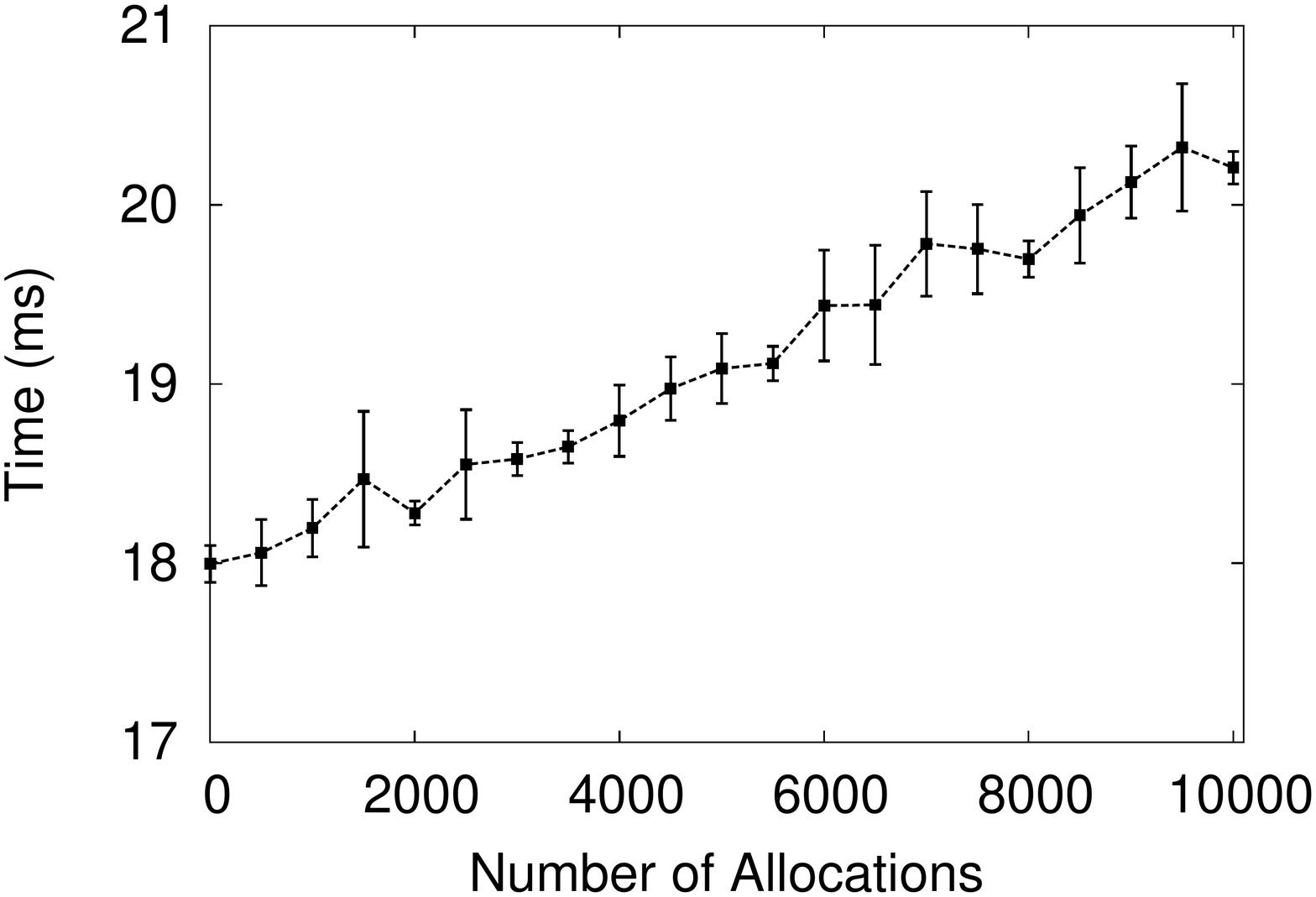} \\
\end{tabular}}
\vspace{-5mm}
\caption{Time taken to verify a delegated policy for an increasing number of
  delegated predicates, increasingly complex regular expressions, and an
  increasing number of bandwidth allocations.}
\label{fig:verify}
\end{figure*}

\paragraph*{Verifying negotiators}
Delegated Merlin policies can be modified by negotiators in three
ways: by changing the predicates, the regular expressions, or the
bandwidth allocations. We ran three experiments to benchmark our
negotiator verification runtime for these cases. First, we increased
the number of additional predicates generated in the delegated
policy. Second, we increased the complexity of the regular expressions
in the delegated policy. The number of nodes in the regular
expression's abstract syntax tree is used as a measure of its
complexity. Finally, we increased the number of bandwidth allocations
in the delegated policy. For all three experiments, we measured the
time needed for negotiators to verify a delegated policy against the
original policy. We report the mean and standard deviation over ten
runs.

The results, shown in Figure \ref{fig:verify}, demonstrate that policy
verification is extremely fast for increasing predicates and
allocations. Both scale linearly up to tens of thousands of
allocations and statements and complete in milliseconds. This shows
that Merlin negotiators can be used to rapidly adjust to changing
traffic loads. Verification of regular expressions has higher
overhead. It scales quadratically, and takes about $3.5$ seconds for
an expression with a thousand nodes in its parse tree. However, since
regular expressions denote paths through the network, it is unlikely
that we will encounter regular expressions with thousands of nodes in
realistic deployments. Moreover, we expect path constraints to change
relatively infrequently compared to bandwidth constraints.

\paragraph*{Dynamic adaptation}
Merlin negotiators support a wide range of resource management
schemes. We implemented two common
approaches: \emph{additive-increase, multiplicative decrease} (AIMD),
and \emph{max-min fair-sharing} (MMFS).  With AIMD, tenants adjust
resource demands by incrementally trying to increasing their
allocation.  With MMFS, tenants declare resource requirements ahead of
time. The negotiator attempts to satisfy demands starting with the
smallest. Remaining bandwidth is distributed among all
tenants. Figure~\ref{fig:adaptation} (a) shows the bandwidth usage
over time for two hosts using the AIMD strategy.
Figure~\ref{fig:adaptation} (b) shows the bandwidth usage over time
for four hosts using the MMFS negotiators. Host \texttt{\small h1}
communicates with \texttt{\small h2}, and \texttt{\small h3}
communicates with \texttt{\small h4}.  Both graphs were generated
using our hardware testbed. Overall, the negotiators allow the network
to quickly adapt to changing resource demands, while respecting the
global constraints imposed by the policy.

\begin{figure}[t]
\centerline{\begin{tabular}{c@{}c}
  \includegraphics[trim=0mm 3.5cm 0mm 0mm, clip=true, width=0.5\columnwidth]{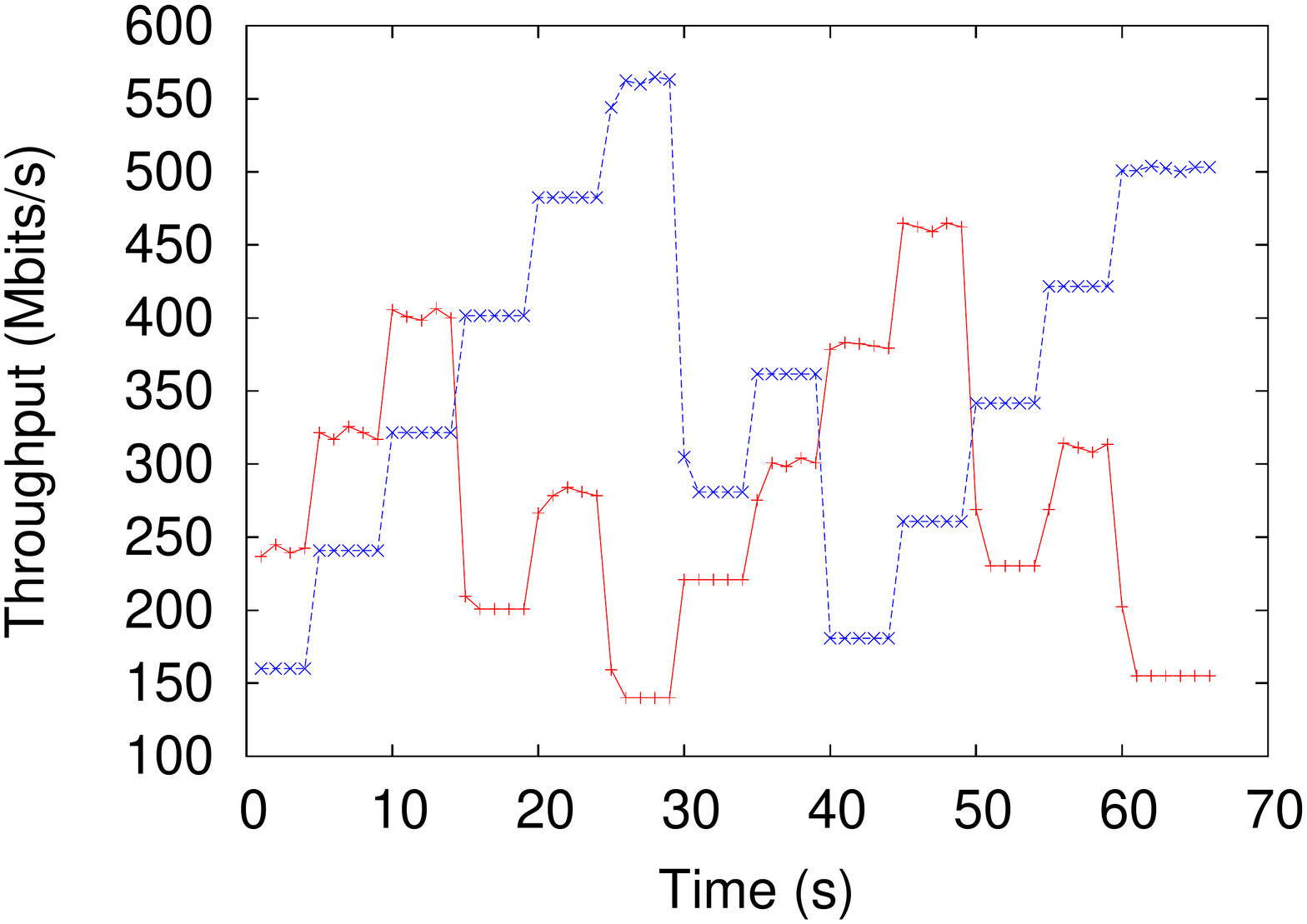} &
  \includegraphics[trim=0mm 3.5cm 0mm 0mm, clip=true, width=0.5\columnwidth]{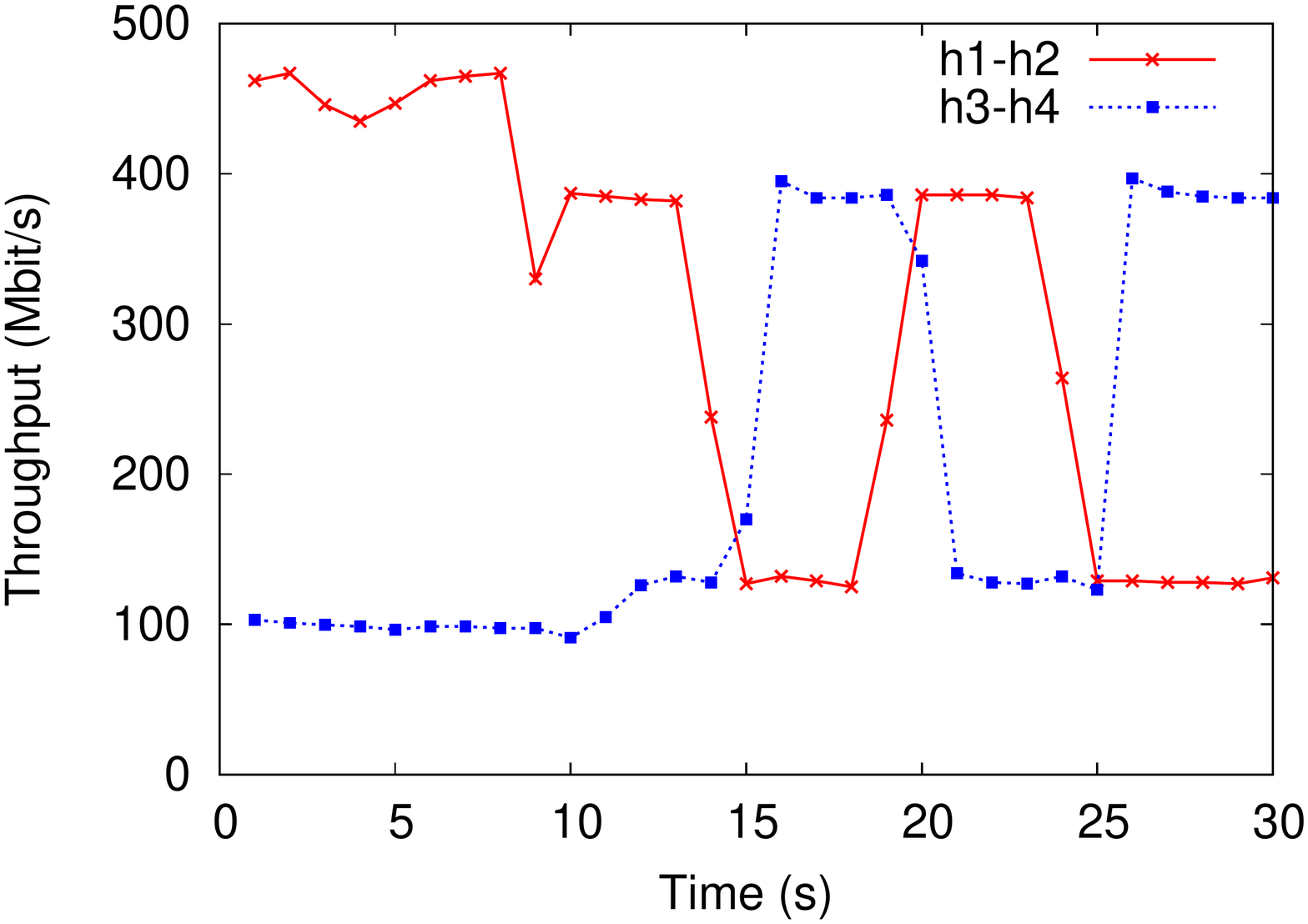} \\
(a) & (b) \\
\end{tabular}}
\vspace{-3mm}
\caption{(a) AIMD and (b) MMFS dynamic adaptation.}
\label{fig:adaptation}
\end{figure}

\section{Related Work}

In prior work~\cite{soule13}, we presented a preliminary design for
Merlin including sketching the encoding of path selection as a
constraint problem, and presenting ideas for language-based delegation
and verification. This paper expands our eariler work with a complete
description of Merlin's design and implementation and an experimental
evaluation.

A number of systems in recent years have investigated mechanisms for
providing bandwidth caps and guarantees~\cite{ballani11,shieh10b,popa12,jeyakumar13}, implementing
traffic filters~\cite{ioannidis00,roesch99}, or specifying forwarding
policies at different points in the
network~\cite{foster11,godfrey09,monsanto12,hinrichs09}. Merlin builds
on these approaches by providing a unified interface and central point
of control for switches, middleboxes, and end hosts.

\textsc{simple}~\cite{qazi13} is a framework for controlling
middleboxes. \textsc{simple} attempts to load balance the network with respect
to \textsc{tcam} and \textsc{cpu} usage. Like Merlin, it solves an optimization
problem, but it does not specify the programming interface to the framework, or
how policies are represented and analyzed. \textsc{aplomb}~\cite{sherry12} is a
system that allows middleboxes to be obtained from a cloud service. Merlin is
designed to provide a programming abstraction for the whole network, of which a
middlebox is just one kind of network element. It would be interesting to extend
Merlin with a back-end that supports cloud-hosted middleboxes as in
\textsc{aplomb}.

Many different programming languages have been proposed in recent
years including Frenetic~\cite{foster11}, Pyretic~\cite{monsanto13},
and Maple~\cite{voellmy13}. These languages typically offer
abstractions for programming OpenFlow networks. However, these
languages are limited in that they do not allow programmers to specify
middlebox functionality, allocate bandwidth, or delegate policies. An
exception is the \textsc{pane}~\cite{ferguson13} system, which allows
end hosts to make explicit requests for network resources like
bandwidth. Unlike Merlin, \textsc{pane} does not provide mechanisms
for partitioning functionality across a variety of devices and
delegation is supported at the level of individual network flows,
rather than entire policies.

The Merlin compiler implements a form of program partitioning. This
idea has been previously used in other domains including secure
web applications~\cite{chong07}, and distributed computing and
storage~\cite{liu09}.

\section{Conclusion}

The success of programmable network platforms has demonstrated the
benefits of high-level languages for managing networks.  Merlin
complements these approaches by further raising the level of
abstraction. Merlin allows administrators to specify the functionality
of an entire network, leaving the low-level configuration of
individual components to the compiler. At the same time, Merlin
provides tenants with the freedom to tailor policies to their
particular needs, while assuring administrators that the global
constraints are correctly enforced. Overall, this approach
significantly simplifies network administration, and lays a solid
foundation for a wide variety of future research on network
programmability.

{\small
\bibliographystyle{abbrv}
\balance
\bibliography{main}
}

\end{document}